\documentclass[journal]{IEEEtran}

\usepackage{blindtext, graphicx}
\usepackage{bm}
\usepackage[font=footnotesize]{caption}
\usepackage[noend]{algorithm2e}
\usepackage{siunitx}
\usepackage{url}

\usepackage{subcaption}

\usepackage{caption,setspace}

\usepackage{booktabs}

\usepackage{localmath}

\usepackage{pifont}
\newcommand{\cmark}{\ding{51}}%
\newcommand{\xmark}{\ding{55}}%

\providecommand{\SINR}{\ensuremath{\mathsf{SINR}}}

\graphicspath{{figures/}{code_figures/}{access_imgs/}{code_figures/speed_contest_20200227-114954_f96a1824ad/}{code_figures/reverb_interf_performance_20200225-200329_c0326397b0/}}

\begin{document}

\title{MM Algorithms for Joint Independent Subspace Analysis with Application to Blind Single and Multi-Source Extraction}

\author{Robin~Scheibler,~\IEEEmembership{Member,~IEEE,}
        and~Nobutaka~Ono,~\IEEEmembership{Senior Member,~IEEE,}
\thanks{Graduate School of Systems Design, Tokyo Metropolitan University, 6-6 Asahigaoka, Hino city, Tokyo, 191-0065 Japan (e-mail: robin.scheibler@ieee.org)}
\thanks{This research was supported by JSPS KAKENHI Grant Numbers JP16H01735 and JST CREST Grant Number JPMJCR19A3.}
\thanks{The software to reproduce the results of this paper is available at \protect\url{https://github.com/fakufaku/jisamm}.}}


\maketitle

\begin{abstract}
  In this work, we propose efficient algorithms for joint independent subspace analysis (JISA), an extension of independent component analysis that deals with parallel mixtures, where not all the components are independent.
We derive an algorithmic framework for JISA based on the majorization-minimization (MM) optimization technique (JISA-MM).
We use a well-known inequality for super-Gaussian sources to derive a surrogate function of the negative log-likelihood of the observed data.
The minimization of this surrogate function leads to a variant of the hybrid exact-approximate diagonalization problem, but where multiple demixing vectors are grouped together.
In the spirit of auxiliary function based independent vector analysis (AuxIVA), we propose several updates that can be applied alternately to one, or jointly to two, groups of demixing vectors.

Recently, blind extraction of one or more sources has gained interest as a reasonable way of exploiting larger microphone arrays to achieve better separation.
In particular, several MM algorithms have been proposed for overdetermined IVA (OverIVA).
By applying JISA-MM, we are not only able to rederive these in a general manner, but also find several new algorithms.
We run extensive numerical experiments to evaluate their performance, and compare it to that of full separation with AuxIVA.
We find that algorithms using pairwise updates of two sources, or of one source and the background have the fastest convergence, and are able to separate target sources quickly and precisely from the background.
In addition, we characterize the performance of all algorithms under a large number of noise, reverberation, and background mismatch conditions.

\end{abstract}

\begin{IEEEkeywords}
Blind source separation, joint independent subspace analysis, overdetermined, majorization-minimization optimization, array signal processing
\end{IEEEkeywords}

%
\IEEEpeerreviewmaketitle

\section{Introduction}
%
%
%
%

\IEEEPARstart{B}{lind} source separation (BSS) is the problem of recovering several signals from one or more of their mixtures, without any side information.
This problem appears in several domains --- among others, audio~\cite{Makino:2018iq}, e.g., for the separation of speech~\cite{Makino:2007vg} and  music~\cite{Cano:2019dw}, biomedical for electrocardiogram~\cite{ZARZOSO:1997dt} and electroencephalogram~\cite{Cong2019}, and digital communications~\cite{Yang:2018kr}.
By far the most popular technique for blind source separation (BSS) is independent component analysis (ICA) which only requires statistical independence of the sources \cite{Comon:1994kr}.
A common variant of vanilla BSS is when there are multiple parallel mixtures whose latent sources have some statistical dependence.
This is known as independent vector analysis (IVA)~\cite{Hiroe:2006ib,Kim:2006ex}, or sometimes joint blind source separation~\cite{YiOuLi:2009gc}.
IVA type of problems are very common when dealing with convolutive mixtures, such as in audio source separation.
By working in the frequency domain, the separation problem can be carried out in parallel on sub-bands~\cite{Smaragdis:1998kl}.
ICA and IVA are applicable for BSS when the number of latent sources is the same as the number of sensors --- the so-called \textit{determined} case.
Independent subspace analysis (ISA) extends ICA to the case where not all latent sources are independent~\cite{Comon:1995vw,DeLathauwer:1995ww,Cardoso:1998cw}.
Instead, each source spans a subspace of dimension possibly larger than one.
Recently, joint independent subspace analysis (JISA) has been proposed to combine the approaches of IVA and ISA~\cite{Silva:ca,Lahat:2016fg,Lahat:2015bp}.

Majorization-minimization (MM), also known as optimization transfer or auxiliary function technique, is a popular way of minimizing challenging objective functions~\cite{Lange:2016wp}.
Rather than directly minimizing the objective, it instead minimizes a \textit{surrogate} function that is both tangent to and majorizes the objective everywhere.
Typically, the surrogate function is chosen to be easier to optimize than the original objective, e.g., smooth and/or with a global, closed-form solution.
As such, MM algorithms are inherently stable and require little to no tuning.

Auxiliary function based ICA and IVA (AuxICA and AuxIVA, respectively) algorithms~\cite{Ono:2010hh,Ono:2011tn,Ono:2020iva} apply this technique to the minimization of the cost functions of ICA and IVA.
AuxIVA is applicable to the separation of super-Gaussian spherical sources, which covers a large number of popular source models (see \cite{Ono:2020iva} for some examples).
The minimization of the surrogate function it uses leads to the so-called hybrid exact-approximate diagonalization (HEAD) problem~\cite{Yeredor:hr,Weiss:2017il}.
While no general solution to this problem is known, AuxIVA uses the \textit{iterative projection} (IP) rules to alternately update the demixing vector associated with each source (IP)~\cite{Ono:2011tn}.
Recently, more efficient joint pairwise updates (IP2) have been proposed~\cite{Ono:2020iva,Ono:2018wk}.
Beyond AuxIVA, the IP rules underpin a large number of separation algorithms using more sophisticated source models, e.g., low-rank~\cite{Kitamura:2016vj}, or deeply learnt~\cite{Makishima:2019fl,Kameoka:2019be}.

Our motivating application is audio convolutive mixture separation.
Traditionally, ICA has been run separately for each frequency sub-band~\cite{Smaragdis:1998kl}, with a permutation alignment step~\cite{Sawada:fk}.
Nevertheless, this extra step is notoriously hard to get right and avoiding it is desirable.
This is where IVA enters the picture as it allows to perform the separation jointly over frequencies, and techniques based on IP form now the state of the art in determined audio BSS~\cite{Ono:2011tn,Ono:2020iva,Kitamura:2016vj,Makishima:2019fl,Kameoka:2019be}.
%
Now, all these techniques are determined, and hence they will attempt at separating as many sources as there are sensors.
It is also accepted wisdom that using more microphones adds robustness and improves performance.
However, there are rarely more than two or three sources simultaneously active and trying to separate more is wasteful.
We thus address the problem of blind source and multi-source extraction (BSE and BMSE, respectively), defined as the recovery of $K$ sound sources recorded with $M$ microphones when $K < M$.  
A straightforward algorithm is to separate $M$ sources, and retain the $K$ outputs with the largest power.
Alternatives to power-based selection exist, for example~\cite{Kitamura:2015fi,Wang:ci,Wang:2018bz}.
Due to the large number of parameters, $\calO(M^2)$, to estimate, such approaches come with a high computational cost.
Ideally, we want to estimate no more than $\calO(KM)$ parameters.

Several methods with better complexities have been proposed.
These methods fall broadly in two categories.
First, some methods can directly tackle BMSE~\cite{Murata:1998ul,Amari:1999jy}, but may require regularization~\cite{Nishikawa:2004ug}.
Second, methods that first reduce the number of channels to $K$ and then apply a determined separation algorithm.
This is done for example by selecting the best $K$ channels~\cite{Nishikawa:2004bn,Osterwise:2014fe}, or by principal component analysis (PCA)~\cite{Osterwise:2014fe,Joho:2000tt,Lee:2007bw}.
Nevertheless, these methods inherently risk removing some target signal upfront, irremediably degrading performance.
Anecdotally, a few methods have been proposed for instantaneous mixtures~\cite{Fu:jx,Souden:2006wf}, and in the time-domain~\cite{Diamantaras:ib}.
All the above methods are single mixture methods that require permutation alignment.
Recently, the single source case has been tackled with a Gaussian background model spanning the subspace not occupied by the target source~\cite{Koldovsky:fn}.
In fact, this model falls squarely in the JISA framework.
In previous work, we extended it to BMSE and proposed the efficient MM-based overdetermined IVA (OverIVA) algorithm~\cite{Scheibler:2019vx}.
For a single source, the MM approach leads to fast independent vector extraction (FIVE)~\cite{Scheibler:2019tt,Ikeshita:2020ic}, which has suprisingly fast convergence.

\textit{Contributions:}
In this work, we derive concrete algorithms for JISA based on the majorization-minimization optimization technique (JISA-MM).
The proposed algorithms extend AuxIVA to the JISA case.
We derive in this paper similar efficient rules to update one, or two, demixing sub-matrices for subspaces, rather than demixing vectors of sources\footnote{The rule for one sub-demixing matrix has been recently independently proposed by Ikeshita et al. in a less general manner~\cite{Ikeshita:2020ic}.}.
We further apply JISA-MM to derive some known~\cite{Scheibler:2019vx,Scheibler:2019tt} and new algorithms for OverIVA.
All the proposed algorithms are hyperparameter-free, guaranteed to decrease the value of the objective function, and converge very quickly.
We run extensive numerical experiments to compare the performance of the different algorithms for the BSE and BMSE of speech signals.
The comparison is made in terms of scale-invariant signal-to-distortion and signal-to-interference ratios (SI-SDR and SI-SIR, respectively)~\cite{LeRoux:2018tq} and convergence speed.
We also study the targets extraction success probability under a large range of signal-to-interference ratios (SINR), background conditions, and reverberation time.
We find that the proposed methods reach performance similar to full separation with AuxIVA, but at a fraction of the computational cost.
They are thus excellent candidates for practical implementations of BSS in multi-microphone systems.

\textit{Organization:}
The rest of this paper is organized as follows.
\sref{background} describes the signal model and gives an introduction to MM algorithms, AuxIVA, and IP.
In \sref{jisa_mm}, we derive the proposed framework for JISA based on MM optimization.
\sref{overiva} describes the BSE and BMSE problems and the different algorithms based on JISA-MM.
The numerical experiments and their results are presented in \sref{perfeval}.
\sref{conclusion} concludes this work.

\section{Background}
\seclabel{background}

We consider the problem of separating $F$ mixtures, each of $K$ components, and recorded using $M$ sensors.
The input signal model is as follows,
\begin{equation}
  x_{mfn} = \sum_{k=1}^K a_{mkf} s_{kfn}, \quad n = 1,\ldots, N,
  \elabel{mixture_model}
\end{equation}
where $x_{mfn}\in \C$ and $s_{kfn}\in\C$ are the $m$th sensor input and $k$th source signals, respectively, in the $f$th mixture, at time $n$.
The mixing coefficient $a_{mkf}\in\C$ controls the amount of the $k$th source in the $m$th sensor signal of the $f$th mixture.
The indices all run between $1$ and the corresponding capital letter, namely $F$, $K$, $M$, and $N$, respectively.

Such parallel mixtures most frequently appear in audio source separation.
Consider the convolutive mixture of $K$ sources recorded by $M$ microphones,
\begin{equation}
	\hat{x}_m[t] = \sum_{k=1}^K (\hat{a}_{mk} \star \hat{s}_k)[t], \quad \forall m =1,\ldots, M,
  \elabel{model_time_domain}
\end{equation}
where $\hat{a}_{mk}[t]$ is the impulse response between the $k$th source and $m$th microphone, and the operator $\star$ denotes convolution.
Then, \eref{mixture_model} is obtained by a time-frequency transformation, typically, the short time Fourier transform (STFT)~\cite{Allen:1977in}, where convolution becomes frequency-wise multiplication.
There, $x_{mfn}$ and $s_{kfn}$ are the STFT of $\hat{x}_m[t]$ and $\hat{s}_k[t]$, respectively, and $a_{mkf}$ is the discrete Fourier transform of $\hat{a}_{mk}[t]$.
The number of mixtures $F$ is the number of positive frequency bins, i.e. $F = \floor{F_{\operatorname{DFT}} / 2} + 1$, where $F_{\operatorname{DFT}}$ is the length of the DFT used.
Note that \eref{mixture_model} is an approximation of \eref{model_time_domain} that only holds if $F$ is sufficiently larger than the true maximum length of the impulse responses $a_{mk}[t]$, $\forall m, k$.

In the rest of the manuscript, we use lower and upper case bold letters for vectors and matrices, respectively.
Furthermore, $\mA^\top$, $\mA^\H$, $\det(\mA)$ and $\tr(\mA)$ denote the transpose, conjugate transpose, determinant, and trace of matrix $\mA$, respectively.
The conjugate of complex scalar $z\in \C$ is denoted $z^*$.
The Euclidean norm of vector $\vv\in\C^d$ is $\| \vv \| = (\vv^\H \vv)^{\frac{1}{2}}$.
Unless specified otherwise, indices $f$, $k$, $m$, and $n$ always take the ranges defined in this section.
As in much of the array signal processing literature, we group all channels in vectors, i.e.,
\begin{align}
  \vx_{fn} & = \begin{bmatrix} x_{1fn} & \cdots & x_{Mfn} \end{bmatrix}^\top, \\
  \vs_{fn} & = \begin{bmatrix} s_{1fn} & \cdots & s_{Kfn} \end{bmatrix}^\top.
\end{align}
At times, we will consider the vector of mixture components of the $k$th source at time $n$
\begin{align}
  \check{\vs}_{kn} = \begin{bmatrix} s_{k1n} & \cdots & s_{kFn} \end{bmatrix}^\top.
  \elabel{src_freq_vec}
\end{align}

\subsection{Independent Vector Analysis of Super-Gaussian Sources}
\seclabel{background_iva}

IVA operates in the determined case where the number of sources and sensors is the same, i.e., $M=K$.
It solves the source separation problem by finding demixing matrices $\mW_f\in \C^{M\times M}$, $f=1, \ldots, F$, such that,
\begin{equation}
  \vy_{fn} = \mW_f \vx_{fn}, \quad \forall f,n,
\end{equation}
are the source estimates and
\begin{equation} 
  \mW_f = \begin{bmatrix} \vw_{1f} & \cdots & \vw_{Mf} \end{bmatrix}^\H.
\end{equation}
IVA posits independence of the sources and a probabilistic model for their content (e.g., spectrograms in audio BSS).
Then, maximum likelihood estimation is used to obtain $\mW_f$, $\forall f$.
A useful model for speech and other temporally non-stationary signals is super-Gaussian~\cite{Benveniste:1990vs,Ono:2010hh}.
Compared to signals following a Gaussian distribution, super-Gaussian signals exhibit a wider range of amplitudes.
While very large components are rare for Gaussian signals, they are fewer, but not rare, for super-Gaussian signals.
We can formalize these requirements into the two following hypotheses.
\begin{hypothesis}[Independence of Sources]
  \label{hyp:independence}
  The separated sources are statistically independent
  \begin{equation}
    \check{\vs}_{kn} \perp \check{\vs}_{k^\prime n^\prime},\ \forall k \neq k^\prime, n, n^\prime,
  \end{equation}
  where $\check{\vs}_{kn}$ and $\check{\vs}_{k^\prime n^\prime}$ are defined as in \eref{src_freq_vec}.
\end{hypothesis}
\begin{hypothesis}[Super-Gaussian Spherical Contrast Function]
  \label{hyp:supergauss}
  The source vectors follow a circular multivariate probablility distribution is circular,
  \begin{equation}
    p_{\vs}(\check{\vs}_{kn}) = \frac{1}{c} e^{-G(\| \check{\vs}_{kn} \|_2)},
  \end{equation}
  where the normalization constant $c$ is independent of $\check{\vs}_{kn}$.
  In addition, $G(r)$ is a real continuous and differentiable function of a real variable $r$ satisfying that $G^\prime(r) / r$ is continuous everywhere and monotonically decreasing on $r \geq 0$ (see \cite{Ono:2020iva,Ono:2010hh} for details).
  This choice includes conventional contrast functions such as $\ell_1$-norm and $\log\cosh$~\cite{Ono:2020iva}.
\end{hypothesis}
Equipped with independence and a source model, it is possible to write explicitely the likelihood function of the data, via a change of variable,
\begin{align}
  \calL(\calW\,;\,\calX) = \prod_f|\det \mW_f|^{2 N} \prod_{kn} p_{\vs}(\check{\vy}_{kn}),
\end{align}
where $\calW = \{\mW_f\}_{f=1}^F$, $\calX = \{\vx_{fn}\}_{f=1, n=1}^{F,N}$, and
\begin{align}
  \check{\vy}_{kn} = \begin{bmatrix} \vw_{k1}^\H \vx_{1n} & \cdots & \vw_{kF}^\H \vx_{Fn} \end{bmatrix}^\top.
\end{align}
Maximiziation of this function is usually carried out via minimization of the negative log-likelihood function
\begin{align}
  \ell(\calW\,;\, \calX) = \sum_{kn} G(\|\check{\vy}_{kn}\|) - 2N\sum_f \log|\det \mW_f|.
  \elabel{cost_iva}
\end{align}
Depending on the choice of $G(r)$, the minimization of this function is non-trivial.
However, for super-Gaussian constrast functions, efficent MM algorithms can be built based on the following inequality~\cite{Ono:2010hh,Ono:2011tn,Ono:2020iva}.
\begin{lemma}[from \cite{Ono:2010hh}]
  \label{lem:supergauss}
  Let $G(r)$ be as defined in Hypothesis~\ref{hyp:supergauss}. Then,
  \begin{equation}
  G(r) \leq G^\prime(r_0)\frac{r^2}{2 r_0} + \left( G(r_0) - \frac{r_0}{2} G^\prime(r_0)\right),
  \end{equation}
  with equality for $r = r_0$.
\end{lemma}
In particular, the IP and IP2 algorithms propose rules to update one~\cite{Ono:2011tn} or two~\cite{Ono:2020iva}, respectively, demixing vectors at a time.


\subsection{Minimization-Maximization Algorithms}

The MM algorithm is a popular optimization technique that allows to tackle non-convex and non-smooth functions~\cite{Lange:2016wp}.
It operates by using a surrogate function that majorizes the objective.
The surrogate function is chosen so that its optimization is easier than the original objective.
The MM algorithm applied to the minimization of the function $f(\vtheta)$ can be summarized as follows.
\begin{proposition}
  \label{prop:mm}
  Let $Q(\vtheta, \hat{\vtheta})$ be a surrogate function such that
  \begin{align}
    Q(\vtheta, \hat{\vtheta}) & \geq f(\vtheta), \quad \forall \vtheta,\hat{\vtheta}, \elabel{surrogate_ineq} \\
    Q(\hat{\vtheta}, \hat{\vtheta}) & = f(\hat{\vtheta}). \elabel{surrogate_eq}
  \end{align}
  Given an initial point $\vtheta_0$, consider the sequence of iterates
  \begin{align}
    \vtheta_{t} = \underset{\vtheta}{\arg\min}\ Q(\vtheta, \vtheta_{t-1}), \quad t=1,\ldots, T.
    \elabel{mm_min_surrogate}
  \end{align}
  Then, the cost function is monotonically decreasing on the sequence, $\vtheta_0, \vtheta_1, ..., \vtheta_T$, i.e.,
  \begin{align}
    f(\vtheta_0) \geq f(\vtheta_1) \geq \ldots \geq f(\vtheta_T).
    \elabel{mm_update}
  \end{align}
\end{proposition}
\begin{proof}
  We simply apply properties of the surrogate and a minimization.
  For any $t=1,\ldots, T$,
  \begin{multline}
    f(\vtheta_{t-1}) = Q(\vtheta_{t-1}, \vtheta_{t-1}) \\
                     \geq \underset{\vtheta}{\min}\ Q(\vtheta, \vtheta_{t-1})
                     = Q(\vtheta_t, \vtheta_{t-1}) \geq f(\vtheta_t),
  \end{multline}
  where we used in order, \eref{surrogate_eq}, \eref{mm_update}, and finally \eref{surrogate_ineq}.
\end{proof}
As can be seen, conveniently little properties of $f$ are required.
In addition, the minimization step of \eref{mm_min_surrogate} can be relaxed to any update that decreases the value of $Q$ without violating Proposition~\ref{prop:mm}.
The MM method has been successfully applied to multi-dimensional scaling~\cite{DeLeeuw:1997vw}, sparse norm minimization as the popular iteratively reweighted least-squares algorithm~\cite{Daubechies:2010hf}, and to IVA~\cite{Ono:2011tn,Ono:2020iva}.

\section{MM Algorithms for Joint Independent Subspace Analysis}
\seclabel{jisa_mm}

ICA and IVA have traditionally operated in the determined regime where the number of independent sources is the same as that of sensors.
We will now consider the generalized model of JISA~\cite{Lahat:2016fg,Lahat:2015bp}, but specialized for super-Gaussian sources.
Then, we will derive efficient updates based on the MM technique for this model.

\subsection{JISA Model for Super-Gaussian Sources}

Let us divide the $M$-dimensional source space into $L$ subspaces.
We denote the index sets of sources belonging to the same subspace by $\calI_1,\ldots,\calI_L$, such that
\begin{equation}
  \bigcup_{\ell=1}^L\calI_\ell=\{1,\dots,M\},\quad \text{and}\quad \calI_\ell \bigcap\calI_{\ell^\prime} = \varnothing,\ \forall \ell \neq \ell^\prime.
\end{equation}
In this case, the parameters to estimate are the $LF$ sub-demixing matrices $\wb{\mW}_{gf}\in \C^{M \times |\calI_g|}$ such that
\begin{equation}
  \vy_{\ell fn} = \wb{\mW}_{\ell f}^\H \vx_{fn}, \quad \forall \ell ,f,
\end{equation}
are the separated subspaces.
The rows of $\wb{\mW}_{\ell f}$ are the demixing vectors $\vw_{kf}^\H$ with $k\in \calI_\ell $, such that
\begin{equation}
  \mW_f = \begin{bmatrix} \wb{\mW}_{1f} & \cdots & \wb{\mW}_{Lf} \end{bmatrix}^\H.
\end{equation}
To simplify the notation, we let $d_\ell = |\calI_\ell|$ for all $\ell$.
Without loss of generality, we take the index of the sources to be ordered in the sets, i.e. $\calI_1 = \{1,\ldots,d_1\}$, $\calI_2=\{d_1+1,\ldots,d_1 + d_2\}$, etc.
We use two hypotheses similar to Hypothesis~\ref{hyp:independence} and Hypothesis~\ref{hyp:supergauss}, but generalized to subspaces.
\begin{hypothesis}[Independence of Subspaces]
  Each subspace is independent from the others, namely,
  \begin{equation}
    \check{\vs}_{kn} \perp \check{\vs}_{k^\prime n^\prime},\ \forall k\in \calI_\ell, k^\prime \not\in \calI_{\ell}.
  \end{equation}
\end{hypothesis}
\begin{hypothesis}[Generalized Spherical Contrast Functions]
  Let $\vs_{\ell fn}$ be the vector whose components are the source signals $s_{kfn}$, for all $k\in\calI_\ell$.
  The probability distribution of the $\ell$th subspace is
  \begin{equation}
    p_\ell(\vs_{\ell 1n},\ldots,\vs_{\ell Fn}) = \frac{1}{c} e^{-G_\ell \left(\sqrt{\sum_f \vs_{\ell fn}^\H \mB_{\ell f}^{-1} \vs_{\ell fn}}\right)},
  \end{equation}
  where $G_\ell(r)$ is a super-Gaussian function as defined in Hypothesis~\ref{hyp:supergauss}.
  The matrices $\mB_{\ell f}$ describe the covariance structure within the $\ell $th subspace in mixture $f$.
  The normalization constant $c$ does not depend on the parameters to estimate.
  Note that the contrast function can be chosen differently for each subspace.
  This can be useful in case prior information of the distribution of the subspaces is available.
\end{hypothesis}
As in IVA, we can write the negative log-likelihood function of the observed data
\begin{multline}
  \calJ(\calW\,;\,\calX) = \sum_{\ell=1}^L \sum_{n=1}^N G_{\ell}\left(\sqrt{\sum_f \vy_{\ell fn}^\H \mB_{\ell f}^{-1} \vy_{\ell fn}}\right) \\
  - 2 N \sum_{f=1}^F \log|\det \mW_f | + \text{constant}.
  \elabel{cost_jisa}
\end{multline}
The inequality from Lemma~\ref{lem:supergauss} yields the surrogate function
\begin{multline}
  \calJ_2(\calW\,;\,\calX) = N \sum_{\ell f} \tr(\wb{\mW}_{\ell f}^\H \mV_{\ell f} \wb{\mW}_{\ell f} \mB_{\ell f}^{-1}) \\
  - 2 N \sum_f \log|\det \mW_f | + \text{constant},
  \elabel{cost_aux_jisa}
\end{multline}
where $\calJ(\calW\,;\,\calX) \leq \calJ_2(\calW\,;\,\calX)$, with the auxiliary variable
\begin{align}
\mV_{\ell f} = \frac{1}{N} \sum_n \varphi_\ell (\bar{r}_{\ell n}) \vx_{fn}\vx_{fn}^\H,
\end{align}
where 
\begin{equation}
  \varphi_\ell (r) = \frac{G_\ell^\prime(r)}{2r},
  \elabel{covmatrix_weight}
\end{equation}
and
\begin{align}
    \bar{r}_{\ell n} = \sqrt{\sum_f \vy_{\ell fn}^\H \mB_{\ell f}^{-1} \vy_{\ell fn}}.
\end{align}
As we will see now, this function can be efficiently optimized with updates similar to those of AuxIVA~\cite{Ono:2010hh,Ono:2011tn,Ono:2020iva}.
In this case too, we can find necessary optimality conditions for the solution.
\begin{proposition}[Necessary Optimality Conditions for \eref{cost_aux_jisa}]
  \label{prop:head_jisa}
  A stationary point of~\eref{cost_aux_jisa} must satisfy the following for all $f$,
  \begin{equation}
    \mW_f \begin{bmatrix} \mV_{1f} \wb{\mW}_{1f} & \cdots & \mV_{Lf} \wb{\mW}_{Lf} \end{bmatrix} =
    \begin{bmatrix}
      \mB_{1f} & \cdots & \vzero \\
      \vdots & \ddots & \vdots \\
      \vzero & \cdots & \mB{Lf}
    \end{bmatrix},
    \elabel{head_jisa}
  \end{equation}
  where the right hand side is a block diagonal matrix with $\mB_{1f},\ldots,\mB_{Lf}$ on the diagonal.
\end{proposition}
\begin{proof}
  The gradient of~\eref{cost_aux_jisa} with respect to $\wb{\mW}_{\ell f}$ is
  \begin{equation}
    \nabla_{\wb{\mW}_{\ell f}^*} \calJ_2 = \mV_{\ell f} \wb{\mW}_{\ell f}\mB_{\ell f}^{-1} - \mW_f^{-1}\mE_\ell ,
  \end{equation}
  where
  \begin{equation}
    \mE_\ell =
    \begin{bmatrix} 
      \vzero_{d_\ell \times \sum_{\ell^\prime=1}^{\ell-1} d_{\ell^\prime}} & \mI_{|\calI_\ell|} & \vzero_{d_\ell\times \sum_{\ell^\prime=\ell+1}^L d_{\ell^\prime}}
    \end{bmatrix}^\top
  \end{equation}
  Setting to zero and rearranging the terms yields the result.
\end{proof}
Equation~\eref{head_jisa} is in fact a special case of the hybrid exact-approximate diagonalization (HEAD) problem~\cite{Yeredor:hr} that also appears in the derivation of AuxIVA~\cite{Ono:2011tn}.
The difference is that in this instance the covariance matrices are shared among the demixing vectors belonging to the same subspace.
As we will see in the next section, this allows to develop more efficient algorithms that jointly update multiple rows of the demixing matrices.

Similarly to the regular HEAD problem, when $L>2$, there is to the best of our knowledge no known general solution to \eref{head_jisa}.
Instead, we propose alternate updates of one sub-demixing matrix, keeping the others fixed.
The special case $L=2$ can be solved globally and applied to joint pairwise updates of two sub-demixing matrices.
These updates can be done in closed form and are derived in a similar way to AuxIVA's IP and IP2, respectively.

\subsection{Update of One Sub-demixing Matrix}
\seclabel{up_1_sub_demix}

We now consider the update of the sub-demixing matrix $\wb{\mW}_{\ell f}$ while keeping $\wb{\mW}_{\ell^\prime f}$, for all $\ell^\prime \neq \ell$, fixed.
That is, we want $\wb{\mW}_{\ell f}$ to be a solution of
\begin{align}
  \underset{\wb{\mW}\in \C^{M\times d_\ell}}{\min}\ \tr(\wb{\mW}^\H \mV_{\ell f} \wb{\mW}) - 2 \log | \det \mW_f |.
  \elabel{cost_aux_jisa_1sub}
\end{align}
We further omit the frequency index $f$ to lighten notation.
\begin{theorem}
  \label{thm:jisa_up1}
  Assume $\mV_\ell$ is full rank and let $\mX$ be an $M\times d_\ell$ matrix with full column rank such that
  \begin{equation}
    \mX^\H \mV_\ell \begin{bmatrix} \wb{\mW}_1 & \cdots & \wb{\mW}_{\ell-1} & \wb{\mW}_{\ell+1} & \cdots & \wb{\mW}_L \end{bmatrix} = 0.
    \elabel{up_1_sub_demix_null}
  \end{equation}
  Further let $\mB_\ell = (\mB_\ell^{1/2})^\H \mB_\ell^{1/2}$, $\mX^\H \mV_\ell \mX = \mQ^\H \mQ$, and $\mR$ be an arbitrary hermitian matrix.
  Then,
  \begin{equation}
    \wb{\mW}_\ell = \mX \mQ^{-1} \mR \mB_\ell^{1/2},
    \elabel{up_1_sub_demix_norm}
  \end{equation}
  globally minimizes \eref{cost_aux_jisa_1sub}.
\end{theorem}
\begin{proof}
  We start by showing verifying the necessary condition of Proposition~\ref{prop:head_jisa} holds for $\wb{\mW}_\ell$.
  By the definition of $\mX$ in \eref{up_1_sub_demix_null} we have
  \begin{equation}
    \wb{\mW}_{\ell^\prime}^\H \mV_\ell \wb{\mW}_\ell = 0,\quad \forall \ell^\prime \neq \ell.
  \end{equation}
  We can easily check the other condition, that is
  \begin{align}
    \wb{\mW}_{\ell}^\H \mV_\ell \wb{\mW}_\ell & = \mB_\ell,
    \elabel{jisa_up1_quad}
  \end{align}
  regardless of $\mR$.

  We will now show that the choice of $\mR$ does not affect the value of the cost function.
  Because of \eref{jisa_up1_quad}, the value of the trace function is constant and we only consider the log-determinant term.
  Since $\det(\mR) = 1$, due to the multiplicative property of determinants, we have
  \begin{multline}
    \det[ \cdots \ \mX \mQ^{-1} \mR \mB_\ell^{1/2} \ \cdots ] = \det[ \cdots \ \mX \mQ^{-1} \mB_\ell^{1/2} \cdots ].
    \nonumber
  \end{multline}

  We are left now to show that the choice of $\mX$ satisfying \eref{up_1_sub_demix_null} does not affect the value of the cost function.
  Let $\mY\neq \mX$ be full column rank and satisfy \eref{up_1_sub_demix_null}
  Because $\mX$ and $\mY$ belong to the same subspace and are full column rank, then $\exists \mA$, invertible, such that $\mY = \mX \mA$.
  Furthermore, we can find $\mQ_2$ such that $\mY^\H \mV_g \mY = \mQ_2^\H \mQ_2$.
  Then, $\exists \mU$ orthonormal and such that $\mQ_2 = \mU \mQ \mA$.
  Thus, $\mY \mQ_2^{-1} = \mX \mQ^{-1} \mU$, and for the same reason as with $\mR$ above, the cost function is unchanged.
  We conclude that any such stationary point is a global minimum of \eref{cost_aux_jisa_1sub}.
\end{proof}
Now there are several ways to choose $\mX$.
One is to apply a QR factorization to the matrix
\begin{equation}
  \mV_\ell \begin{bmatrix} \wb{\mW}_1 & \cdots & \wb{\mW}_{\ell-1} & \wb{\mW}_{\ell+1} & \cdots & \wb{\mW}_L \end{bmatrix}.
\end{equation}
However, we use a trick similar to AuxIVA.
Let $\mU$ be of size $M\times d_g$ and such that the following matrix is invertible,
\begin{equation}
  \wh{\mW} = \begin{bmatrix} \wb{\mW}_1 & \cdots & \wb{\mW}_{\ell-1} & \mU & \wb{\mW}_{\ell+1} & \cdots & \wb{\mW}_L \end{bmatrix}^\H.
\end{equation}
The value of $\wb{\mW}_\ell$ at the previous iterate is a good choice for $\mU$ in practice.
Then, we choose
\begin{equation}
  \mX = (\wh{\mW} \mV_\ell)^{-1} \mE_\ell.
\end{equation}
One can easily check that \eref{up_1_sub_demix_null} is satisfied, and we found it to be faster than QR decomposition in practice.

\subsection{The Special Case of Two Subspaces}
\seclabel{two_subspaces}

The case of two subspaces turns out to be special, and the surrogate function \eref{cost_aux_jisa} can be minimized globally.
\begin{theorem}[JISA with $L=2$ subspaces]
  \label{thm:jisa_L2}
  Consider \eref{cost_aux_jisa} with $L=2$, and where $\mV_1$ and $\mV_2$ are full rank, positive definite matrices.
  Let $\lambda_1 \geq \ldots \geq \lambda_M$, and $\vu_1, \ldots, \vu_M$ be the eigenvalues and eigenvectors, respectively, of the eigenvalue problem
  \begin{equation}
    \mQ^{-\H} \mV_1 \mQ^{-1} \vu = \lambda \vu,
  \end{equation}
  where $\mQ$ is a square matrix such that $\mQ^\H \mQ = \mV_2$.
  Let in addition $\mB_\ell^{1/2}$ be such that $(\mB_\ell^{1/2})^\H \mB_\ell^{1/2} = \mB_\ell$, for $\ell=1,2$.
  Then, a global minimum of \eref{cost_aux_jisa} is attained for
  \begin{align}
    \wb{\mW}_1 & = \mQ^{-1} \mU_1 \mR_1 \mB_1^{1/2}, \\
    \wb{\mW}_2 & = \mQ^{-1} \mU_2 \mR_2 \mB_2^{1/2},
    \elabel{jisa_L2_sol}
  \end{align}
  where
  \begin{align}
    \mU_1 & = \begin{bmatrix} \vu_{M} & \cdots & \vu_{M - d_1} \end{bmatrix} \mD^{-\frac{1}{2}}, \\
    \mU_2 & = \begin{bmatrix} \vu_{1} & \cdots & \vu_{M - d_1 - 1} \end{bmatrix},
  \end{align}
  with $\mD = \diag(\lambda_M,\ldots,\lambda_{M-d_1})$,
  and $\mR_1,\mR_2$ are arbitrary hermitian matrices.
\end{theorem}
\begin{proof}
  We first prove that the necessary condition from Proposition~\ref{prop:head_jisa} is satisfied.
  Due to the properties of the eigenvectors and the definition of $\mQ$,
  \begin{align}
    \mU_2^\H \mQ^{-\H} \mV_1 \mQ^{-1} \mU_1 & = \mU_2^\H \mU_1 = 0, \\
    \mU_1^\H \mQ^{-\H} \mV_2 \mQ^{-1} \mU_2 & = \mU_1^\H \mU_2 = 0,
  \end{align}
  and thus $\wb{\mW}_2^\H \mV_1 \wb{\mW}_1 = \wb{\mW}_1^\H \mV_2 \wb{\mW}_2 = 0$.
  Furthermore,
  \begin{align}
    \mU_1^\H \mQ^{-\H} \mV_1 \mQ^{-1} \mU_1 & = \mU_1^\H \mU_1 \mD = \mD, \\
    \mU_2^\H \mQ^{-\H} \mV_2 \mQ^{-1} \mU_2 & = \mU_2^\H \mU_2 = \mI,
  \end{align}
  so that $\wb{\mW}_\ell^\H \mV_\ell \wb{\mW}_2 = \mB_\ell$, for $\ell=1,2$.

  Now it is clear that the necessary conditions are fulfilled for any assignment of the eigenvectors to the columns of $\mU_1$ and $\mU_2$.
  Ignoring the hermitian matrices $\mR_1$ and $\mR_2$, each of the $\binom{n}{k}$ possible assignments corresponds to a stationary point of the surrogate function.
  We will show that swapping columns between $\mU_1$ and $\mU_2$ leads to a larger value of \eref{cost_aux_jisa}.
  Because the trace function is constant under the choice \eref{jisa_L2_sol}, we concentrate on the determinant.
  First, we can factorize the demixing matrix as follows
  \begin{equation}
    \mW^\H = \mQ^{-1}
    \begin{bmatrix} \mU_1 & \mU_2 \end{bmatrix}
    \begin{bmatrix} \mR_1 & \vzero \\ \vzero & \mR_2 \end{bmatrix}
    \begin{bmatrix} \mB_1^{1/2} & \vzero \\ \vzero & \mB_2^{1/2} \end{bmatrix}.
  \end{equation}
  Because the left-most and right-most terms are constant, they do not affect the value of the cost function.
  The second term from the right is a hermitian matrix and has unit determinant.
  We thus concentrate on $\begin{bmatrix} \mU_1 & \mU_2 \end{bmatrix}$.
  Without loss of generality, we will swap $\vu_1$ with $\vu_M/\sqrt{\lambda_M}$.
  We have the following equality,
  \begin{equation}
    \left| \det \begin{bmatrix} \frac{\vu_M}{\sqrt{\lambda_M}} & \vu_1 & \cdots \end{bmatrix} \right|
    = \sqrt{\frac{\lambda_1}{\lambda_M}} \left|\det \begin{bmatrix} \frac{\vu_1}{\sqrt{\lambda_1}} & \vu_M & \cdots \end{bmatrix}\right|
  \end{equation}
  and since $\frac{\lambda_1}{\lambda_M} \geq 1$, the proof follows.
\end{proof}
A less general version of this theorem was proved in~\cite{Scheibler:2019tt}.

\subsection{Pairwise Update of Two Sub-demixing Matrices}
\seclabel{jisa_update_2sub}

When there are more than two subspaces, it is still possible to exploit the result of the previous section to perform joint updates of pairs of sub-demixing matrices.
Without loss of generality, we consider an update of the pair $\wb{\mW}_1, \wb{\mW}_2$,
\begin{equation}
  \underset{\wb{\mW}_1\in\C^{M\times d_1} \atop \wb{\mW}_2\in\C^{M\times d_2}}{\min}\ \sum_{\ell\in\{1,2\}} \tr\left(\wb{\mW}_\ell^\H \mV_\ell\wb{\mW}_\ell\mB_\ell^{-1}\right) - 2 \log |\det \mW |.
  \elabel{cost_aux_jisa_2sub}
\end{equation}
In contrast with the two-subspace problem of previous section, $\mW$ has extra rows containing the fixed sub-demixing matrices $\wb{\mW}_3^\H,\ldots,\wb{\mW}_L^\H$.
\begin{theorem}
  \label{thm:jisamm_three_plus}
  Let $\mX$ be an $M \times (d_1 + d_2)$ full-rank matrix such that
  \begin{align}
    \mX^\H \begin{bmatrix}\wb{\mW}_3 & \cdots & \wb{\mW}_L \end{bmatrix} = 0,
  \end{align}
  and let $\mP_\ell = \mV_\ell^{-1}\mX$, $\ell=1,2$.
  Define $\wt{\mV}_\ell = \mP_\ell^\H \mV_\ell \mP_\ell$ and let $\wb{\mU}_1,\wb{\mU}_2$ be a solution to the two-subspace problem with matrices $\wt{\mV}_1^{-1},\wt{\mV}_2^{-1}$.
  Then, $\wb{\mW}_\ell = \mP_\ell \wt{\mV}_\ell^{-1} \wb{\mU}_\ell$, $\ell=1,2$, is a global minimizer of \eref{cost_aux_jisa_2sub}.
\end{theorem}
\begin{proof}
  From the choice of $\mP_\ell$, $\ell=1,2$, it is clear that the necessary conditions,
  \begin{equation}
    \wb{\mW}_{\ell^\prime}^\H \mV_\ell \wb{\mW}_\ell = 0,\quad \forall \ell=1,2,\ \ell^\prime \geq 3,
  \end{equation}
  are satisfied.
  We will now introduce the parametrized $\wb{\mW}_\ell=\mP_\ell \wt{\mU}_\ell$ into~\eref{cost_aux_jisa}.
  Ignoring all constant terms, we have
  \begin{multline}
    \sum_{\ell=1,2} \tr \left(\wt{\mU}_\ell^\H \wt{\mV}_\ell \wt{\mU}_\ell \mB_\ell^{-1}\right) \\
    - 2 \log \left| \det\begin{bmatrix} \mP_1 \wt{\mU}_1 & \mP_2 \wt{\mU}_2 & \wt{\mW} \end{bmatrix} \right|,
    \elabel{subst_cost}
  \end{multline}
  where $\wt{\mW} = \begin{bmatrix}\wb{\mW}_3 & \cdots & \wb{\mW}_L \end{bmatrix}$.
  Now, consider the square invertible matrix $\mA = \begin{bmatrix} \mX & \wt{\mW} \end{bmatrix}^\H$.
  Applying $\mA$ and its inverse to the term in the determinant does not change the value of the cost function,
  \begin{multline}
    \det \left(\mA^{-1} \mA \begin{bmatrix} \mP_1 \wt{\mU}_1 & \mP_2 \wt{\mU}_2 & \wt{\mW} \end{bmatrix}\right) \\
    = \det\left(\mA^{-1}\begin{bmatrix} \mX^H \mP_1 \wt{\mU}_1 & \mX^\H \mP_2 \wt{\mU}_2 & \vzero \\ \wt{\mW}^H \mP_1 \wt{\mU}_1 & \wt{\mW}^\H \mP_2 \wt{\mU}_2 & \wt{\mW}^\H \wt{\mW} \end{bmatrix}\right)\\
    = \det(\mA^{-1})\, \det\begin{bmatrix} \wt{\mV}_1 \wt{\mU}_1 & \wt{\mV}_2 \wt{\mU}_2 \end{bmatrix}\, \det(\wt{\mW}^\H \wt{\mW}).
    \elabel{block_matrix}
  \end{multline}
  where the second equality is due to block diagonality, and the fact that, for $\ell=1,2$,
  \begin{equation}
    \mX^\H \mP_\ell = \mX^\H \mV_\ell^{-1} \mX = \mP_\ell^\H \mV_\ell \mP_\ell = \wt{\mV}_\ell.
  \end{equation}
  Finally, we use the substitution $\wb{\mU}_\ell=\wt{\mV}_\ell \wt{\mU}_\ell$, $\ell=1,2$, and replace in~\eref{subst_cost}.
  Up to a constant term, we obtain the cost function of the two subspace problem of \sref{two_subspaces}.
  Applying Theorem~\ref{thm:jisa_L2} and back-substituting yields the result.
\end{proof}
Now, by applying a little algebra, we can make this result more amenable to implementation.
\begin{corollary}
  Let $\wt{\mV}_2 = \mQ^\H \mQ$.
  Further let $\hat{\lambda}_1\geq\ldots \geq \hat{\lambda}_{d_1+d_2}$ and $\hat{\vu}_1,\ldots,\hat{\vu}_{d_1+d_2}$ be the eigenvalues and eigenvectors, respectively, of $\mQ^{-\H}\wt{\mV}_1 \mQ^{-1}$.
  Then, the solution to \eref{cost_aux_jisa_2sub} is
  \begin{align}
    \wb{\mW}_1 & = \mQ^{-1} [\vu_1, \ldots,\vu_{d_1}] \diag(\hat{\lambda}_1,\ldots,\hat{\lambda}_{d_1})^{-1/2} \mR_1 \mB_1^{1/2}, \nonumber \\
    \wb{\mW}_2 & = \mQ^{-1} [\vu_{d_1+1}, \ldots,\vu_{d_1+d_2}] \mR_2 \mB_2^{1/2}, \nonumber
  \end{align}
  for any hermitian matrices $\mR_1,\mR_2$.
\end{corollary}
\begin{proof}
  The proof follows from plugging the solutions given by Theorem~\ref{thm:jisa_L2} for $\wt{\mV}_1^{-1},\wt{\mV}_2^{-1}$ into Theorem~\ref{thm:jisamm_three_plus} and using the properties of eigenvalues of inverse matrices.
\end{proof}
Interestingly, this results in formulas similar to those of Theorem~\ref{thm:jisa_L2}, but with the order of the eigenvalues reversed.

\section{Blind Single and Multi-Source Extraction}
\seclabel{overiva}

We are now ready to apply the JISA-MM framework to BSE and BMSE.
We consider the overdetermined IVA model proposed in~\cite{Scheibler:2019vx}, which is an extension of independent vector extraction (IVE)~\cite{Koldovsky:fn} to multiple sources.
The sensor signals $\vx_{fn}$ of mixture $f$ at time $n$ is modelled as
\begin{equation}
  \vx_{fn} = \mA_f \vs_{fn} + \mPsi_f \vz_{fn} + \vb_{fn},
\end{equation}
where $\vs_{fn} = [s_{1fn},\ldots,s_{Kfn}]^\top \in \C^K$ contains the source signals, $\vz_{fn} \in \C^{M-K}$ is a vector of coherent noise sources, and $\mA_f \in \C^{M\times K}$ and $\mPsi_f \in \C^{M \times M-K}$ are their respective mixing matrices.
The term $\vb_{fn}\in\C^M$ is an uncorrelated noise vector.
Because the background can be considered as a subspace of sources that do not need to be separated, we can apply the JISA-MM framework developed in the previous section.
We use $K + 1$ subspaces.
Each of the first $K$ contains exactly one source, while the $(K+1)$th contains all the remaining background components, i.e.,
\begin{align}
  \calI_g & = \{g\}, & g&=1,\ldots,K, \\
  \calI_{K+1} & = \{K+1,\ldots, M\}. & &
\end{align}
Now, we would like to ensure that the target sources and background components are correctly attributed to the subspaces.
This is achieved here by using different probabilistic models to sources and background components.
Target sources are non-Gaussian, and the background is described by a time-invariant Gaussian distribution.
The intuition is that background components, being mixtures of many components, will tend to be Gaussian by virtue of the central limit theorem.
In this work, we use the super-Gaussian model of Hypothesis~\ref{hyp:supergauss} for target sources.
We formalize the background model in the following hypothesis.
\begin{hypothesis}[Distribution of Noise Sources]
  \label{hyp:background}
  The separated background noise vectors have a \textit{time-invariant} complex Gaussian
  distribution across sensors
  \begin{equation}
    p_{\vz_f}(\vz_{fn}) = \frac{1}{\pi^{M-K} |det(\mB_{f})|} e^{-\vz_{fn}^H (\mB_{f})^{-1} \vz_{fn}}
    \elabel{dist_noise}
  \end{equation}
  where $\mB_{f}$ is the (unknown) spatial covariance matrix of the noise (after separation).
  Moreover, the separated background noise is statistically independent across frequencies.
\end{hypothesis}
Based on this model, we want to recover the demixing matrices $\mW_f$, $f=1,\ldots,F$,
\begin{equation}
  \mW_f = \begin{bmatrix} \vw_{1f} & \cdots & \vw_{Kf} & \mU_f \end{bmatrix}^\H,
  \elabel{overiva_demix_struct}
\end{equation}
where $\vw_{kf}$ are demixing vectors as described in \sref{background_iva} and $\mU_f$ is the sub-demixing matrix extracting the background.
We are now ready to write the cost function corresponding to Hypotheses~\ref{hyp:independence},~\ref{hyp:supergauss}, and~\ref{hyp:background},
\begin{multline}
  \calO(\calW\,;\,\calX) = \sum_{kn} G(\| \check{\vy}_{kn} \|) + \sum_{fn} \vx_{fn}^\H \mU_f \mB^{-1}_f \mU_f^\H \vx_{fn} \\
  -2N\sum_f \log|\det(\mW_f)| + \text{constant}.
  \elabel{cost_overiva}
\end{multline}
Using the inequality of Lemma~\ref{lem:supergauss}, we obtain the following surrogate function
\begin{multline}
  \calO_2(\calW\,;\,\calX) = -2N\sum_f \log|\det(\mW_f)| \\
 + N \sum_{kf} \vw_{kf} \mV_{kf} \vw_{kf} + N \sum_{f} \tr\left(\mU_f^\H \mC_f \mU_f \mB^{-1}_f\right) \\
 + \text{constant},
  \elabel{cost_aux_overiva}
\end{multline}
such that $\calO(\calW\,;\,\calX) \leq \calO_2(\calW\,;\,\calX)$.
Using $\calO_2$, we can build several MM algorithms, that we describe in the rest of this section.
It is clear that conventional IP and IP2 rules can be applied to update the demixing vectors one-by-one and two-by-two, respectively.
The background subspace can be updated also individually by using Theorem~\ref{thm:jisa_up1}.
In addition, it is possible to update jointly one demixing vector and the background using Theorem~\ref{thm:jisamm_three_plus}.
This last scheme leads to the particularly efficient algorithm FIVE in the single target source case~\cite{Scheibler:2019tt}.

\subsection{Estimation of the Covariance of the Background}

The model proposed in Hypothesis~\ref{hyp:background} posits a covariance matrix $\mB_f$ for the background that is often unknown in practice.
While we omit the details here, all other parameters being fixed, the maximum likelihood estimator or the covariance matrix is
\begin{equation}
  \wh{\mB}_f = \mU_f^\H \mC_f \mU_f.
  \elabel{noise_covmat_est}
\end{equation}
Now recall that all of the updates presented in \sref{jisa_mm} allow to fit arbitrary covariance structures to the subspaces.
This means that we do not really need to know $\mB_f$ in advance.
We can set it arbitrarily, e.g., $\mB_f = \mI$, or to a value that saves computations.

\subsection{Background Update with Parametrized Demixing Matrix}
\seclabel{param_bg_update}

Let us consider again the problem at hand.
Our objective is to estimate the demixing matrix $\mW_f$ such that the source vector $\vs_{fn}$ is recovered from the measurements
\begin{equation}
    \begin{bmatrix} \vs_{fn} \\ \mPhi_f \vz_{fn} \end{bmatrix} = \wh{\mW}_f \vx_{fn}.
\end{equation}
The matrix $\mPhi_f$ is an arbitrary invertible linear transformation reflecting that we do not aim at separating the background noise components.
Indeed, we may even choose $\mPhi_f$ to simplify the task at hand.
Namely, we choose it so that
\begin{equation}
  \mU_f = \begin{bmatrix} \mJ_f \\ -\mI_{M-K} \end{bmatrix}.
  \elabel{background_param}
\end{equation}
with $\mJ_f\in\C^{M-K\times K}$.
With a slight abuse of notation, we let $\vz_{fn} = \mU_f \vx_{fn}$.
In previous work, an orthogonal constraint has been introduced with little justification to help with the estimation of $\mU_f$~\cite{Koldovsky:fn,Scheibler:2019vx}, namely,
\begin{equation}
  \Expect{\vz_{fn} \vs_{fn}^\H} = \mU_f^\H \mC_f \wt{\mW}_f = \vzero.
  \elabel{orth_const}
\end{equation}
where $\wt{\mW}_f = \begin{bmatrix} \vw_{1f} & \cdots & \vw_{Kf} \end{bmatrix}$.
In light of the development in the previous section, it is now clear that~\eref{orth_const} is a necessary condition for the optimality of $\mU_f$ as stated in Proposition~\ref{prop:head_jisa}.
In fact, \eref{orth_const} appears in the update rules for one sub-demixing matrix as \eref{up_1_sub_demix_null}.
This suggests the following procedure to update $\mU_f$.
For fixed $\wt{\mW}_f$, we can solve \eref{orth_const} for $\mJ_f$ and obtain
\begin{equation}
  \mJ_f = \left( \mE_2 \mC_f \wt{\mW}_f^H \right) \left( \mE_1 \mC_f \wt{\mW}_f^H \right)^{-1},
  \elabel{J_up}
\end{equation}
where $\mE_1 = [\mI_K\ \vzero_{K\times M-K}]$ and $\mE_2 = [\vzero_{M-K\times K}\ \mI_{M-K}]$.
The complexity of this update is dominated by the inversion of a $K\times K$ matrix.
It is thus relatively cheap and can be done after every update of a demixing filter.

Finally, note that $\mU_f$ is fully determined by the current value of the demixing filters.
As a consequence, so is the sample covariance of the background, i.e., $\mB_f$, which is given by~\eref{noise_covmat_est}.

\subsection{Background Update without Parametrization}
\seclabel{non_param_bg_update}

Rather than using the above parametrization, one could apply the update of Thoerem~\ref{thm:jisa_up1} to the background matrix.
Unlike the parametrized case from previous section, \eref{up_1_sub_demix_null} does not fully specify $\mU_f$.
It is possible to use the extra degrees of freedom to fit a covariance structure by applying~\eref{up_1_sub_demix_norm}, because ultimately it does not matter, we would rather omit this step.
However, completely omitting this step would leave the scale of $\mU_f$ ambiguous and might lead to numerical problems in the algorithm.
Instead, we propose to only constrain the diagonal elements of $\mB_f$ to be one.
This can be achieved by normalizing the columns of $\mU_f$,
\begin{equation}
  \vu_{kf} \gets \frac{\vu_{kf}}{\sqrt{\vu_{kf}^\H \mC_f \vu_{kf}}}, \quad k=1, \ldots, M-K,
\end{equation}
where $\vu_{kf}$ is the $k$th column of $\mU_f$.
Now this update is more computationally demanding than with the parametrized matrix.
It requires the inversion of an $M\times M$ matrix.

\subsection{Joint Update of One Demixing Vector and Background}

Rather than having alternate updates for demixing vector and background, it is appealing to have a joint update.
This is achieved by minimizing the surrogate function~\eref{cost_aux_overiva} jointly for the two.
The rules to update the vectors can be directly adapted from \sref{jisa_update_2sub}.
Here, the covariance matrix $\mB_f$ can be chosen according to either of the strategies presented in the last two sections.
However, eigenvalue solvers often yield solutions that satisfy $\mB_f=\mI$, such that no extra processing is required.

When there is only a single target source, then the special case of two subspaces described in \sref{two_subspaces} applies.
This leads to the so-called FIVE algorithm~\cite{Scheibler:2019tt}.
This algorithm applies the updates of Theorem~\ref{thm:jisa_L2} to matrices $\mV_{1f}$ and $\mC_f$.
Because $\mC_f$ is never updated a few tricks can be used to lower the computation and memory requirements.
For example, the background matrix is never needed and does not need to be stored.
In addition, the algorithms has an intuitive interpretation as beamforming iteratively maximizing the SINR.
It is blazingly fast and converges in just a few iterations as demonstrated in \sref{perfeval}.

%

\begin{algorithm}[t]
\SetKwInOut{Input}{Input}\SetKwInOut{Output}{Output}
\SetKw{KwBy}{by}
\Input{Microphones signals $\vx_{fn}\in\C^M$, $\forall f,n$}
\Output{Separated signals $\vy_{fn}\in\C^K$, $\forall f,n, K < M$}
\DontPrintSemicolon
$\mW_f \gets \mI_K,\ \forall f$\;
$\vy_{fn} \gets \vx_{fn},\ \forall f, n$\;
$\mC_f \gets \frac{1}{N}\sum_f \vx_{fn} \vx_{fn}^\H$\;
\For{loop $\leftarrow 1$ \KwTo $\text{max. iterations}$ \KwBy $2$}{
  \For{$k \leftarrow 1$ \KwTo $2K$ \KwBy $2$}{
    $r_{qn} \gets \frac{1}{F} \sum_f |y_{kfn}|^2,\ \forall n, \forall q=k, k+1$\;
    \For{$f \gets 1$ \KwTo $F$}{
      \nl \# Background update, $\mU_f$ as in \eref{overiva_demix_struct}
      $\mU_f \gets (\mW_f \mC_f)^{1} [\ve_{K+1}\,\cdots\, \ve_M]$\;
      \For{$k\gets1$ \KwTo $M-K$}{
        $\vu_{kf} \gets \frac{\vu_{kf}}{\sqrt{\vu_{kf}^\H \mC_f \vu_{kf}^\H}}$\;
      }
      \nl \# Pairwise source update\;
      \For{$q \gets k, k+1$}{
        $\mV_{qf} \gets \frac{1}{N} \sum_n \frac{1}{r_{kn}} \vx_{fn} \vx_{fn}^H$\;
        $\mP_{qf} \gets (\mW_f \mV_{qf})^{-1} [ \ve_k \, \ve_{k+1} ]$\;
        $\wt{\mV}_{qf} \gets \mP_{qf}^\H \mV_{qf} \mP_{qf}$\;
      }
      Let $\vh_k, \vh_{k+1}$ and $\lambda_k \geq \lambda_{k+1}$ be the two eigenvectors and values, respectively, of $\wt{\mV}_{(k+1)f}^{-1}\wt{\mV}_{kf}$\;
      \For{$q \gets k, k+1$}{
        $\vw_{qf} \gets \frac{\mP_{qf} \vh_q}{\sqrt{\vh_q^\H \wt{\mV}_{qf} \vh_q}}$\;
        $y_{qfn} \gets \vw_{qf} \vx_{fn},\ \forall n$\;
      }
    }
  }
}
\vspace{0.5cm}
\caption{OverIVA-IP2-NP: OverIVA with iterative projection 2 and non-parametric background update}
\label{alg:overiva-ip2-np}
\end{algorithm}

\begin{algorithm}[t]
\SetKwInOut{Input}{Input}\SetKwInOut{Output}{Output}
\SetKw{KwBy}{by}
\Input{Microphones signals $\vx_{fn}\in\C^M$, $\forall f,n$}
\Output{Separated signals $\vy_{fn}\in\C^K$, $\forall f,n, K < M$}
\DontPrintSemicolon
$\mW_f \gets \mI_K,\ \forall f$\;
$\vy_{fn} \gets \vx_{fn},\ \forall f, n$\;
$\mC_f \gets \frac{1}{N}\sum_f \vx_{fn} \vx_{fn}^\H$\;
\For{loop $\leftarrow 1$ \KwTo $\text{max. iterations}$}{
  \For{$k \leftarrow 1$ \KwTo $K$}{
    $r_{kn} \gets \frac{1}{F} \sum_f |y_{kfn}|^2,\ \forall n$\;
    \For{$f \gets 1$ \KwTo $F$}{
      $\mV_{kf} \gets \frac{1}{N} \sum_n \frac{1}{r_{kn}} \vx_{fn} \vx_{fn}^H$\;
      $\mP_{kf} \gets (\mW_f \mV_{kf})^{-1} [ \ve_k \, \ve_{K+1}\,\cdots\,\ve_M ]$\;
      $\mR_{f} \gets (\mW_f \mC_{f})^{-1} [ \ve_k \, \ve_{K+1}\,\cdots\,\ve_M ]$\;
      $\wt{\mV}_{kf} \gets \mP_{kf}^\H \mV_{kf} \mP_{kf}$\;
      $\wt{\mC}_f \gets \mR_{f}^\H \mC_f \mR_f$\;
      Let $\vh_1, \ldots, \vh_{M-K+1}$ and $\lambda_1 \geq \ldots \geq \lambda_{M-K+1}$ be the eigenvectors and values, respectively, of $\wt{\mC}_{f}^{-1}\wt{\mV}_{kf}$\;
      $\vw_{kf} \gets \frac{\mP_{kf} \vh_1}{\sqrt{\vh_1^\H \wt{\mV}_{kf} \vh_1}}$\;
      \For{$q \gets 1$ \KwTo $M-K$}{
        $\vw_{(K+q)f} \gets \frac{\mR_f \vh_{q+1}}{ \sqrt{\vh_{q+1}^\H \wt{\mC}_f \vh_{q+1}}}$\;
      }
      $y_{kfn} \gets \vw_{kf} \vx_{fn},\ \forall n$\;
    }
  }
}
\vspace{0.5cm}
\caption{OverIVA-DX/BG: OverIVA with joint updates of one demixing vector and background}
\label{alg:overiva-dx/bg}
\end{algorithm}

\subsection{Computational Complexity}
\seclabel{comp}

For all algorithms described in this section, when the number of time frames $N$ is larger than the number of microphones $M$, the runtime is dominated by the computation of the weighted covariance matrix $\mV_{kf}$.
The computational complexity in that case is $\calO(K F M^2 N)$. When the number of microphones is larger, the bottleneck is either a matrix inversion or an eigenvalue problem.
In both cases, the complexity is $\calO(K F M^3)$.
The total complexity of the algorithms is thus
\begin{equation}
  \calC_{\text{OverIVA}} = \calO(KF M^2 \max\{M, N\}).
\end{equation}
The leading $K$ comes from the number of demixing filters (one per source), and $F$ is the number of frequency bins.
In contrast, conventional AuxIVA needs to update all $M$ demixing filters, which leads to complexity
\begin{equation}
  \calC_{\text{AuxIVA}} = \calO(F M^3 \max\{M, N\}).
\end{equation}
The overall complexity is thus reduced by a factor $K/M$.
This is significant in many practical cases as the number of target sources is rarely larger than three, and the number of microphones can easily be over ten for larger arrays.

\begin{figure}
\begin{center}
  \includegraphics{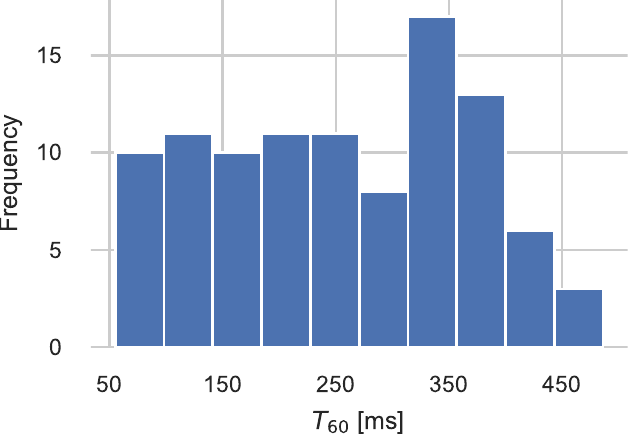}
\end{center}
\caption{Histogram of the $T_{60}$ of the 100 rooms in experiment 1.}
\flabel{experiment_rt60}
\end{figure}

\subsection{Discussion}
\seclabel{discussion}

A few points are in order. We assume the covariance matrix of the noise is rank $M-K$. In practice, this means that we will not be able to remove noise that has the same steering vector as one of the sources.
Independence of noise across frequencies is a simplifying assumption and is typically not fulfilled.
We confirm in the experiment of \sref{perfeval} that this does not seem to be a problem.

One can also wonder how the algorithm can tell apart sources from noise.
While we do not offer a precise analysis, we conjecture that the $K$ strongest sources have a very non-Gaussian distribution.
On the contrary, the mix of the noise and remaining weaker sources will have a distribution closer to Gaussian.
As such, we expect the maximum likelihood to choose the strongest sources automatically.

\begin{table*}
  \centering
  \begin{tabular}{@{}lcccccl@{}}
    \toprule
    \textbf{Label} & \multicolumn{2}{c}{\textbf{Target}} & \textbf{Cost function} & \multicolumn{2}{c}{\textbf{Updates}} & \textbf{Note} \\
    \midrule
    & {Single} & {Multi} & & {Targets} & {Background} & \\
    \cmidrule{2-3}\cmidrule{5-6}
    OverIVA-IP~\cite{Scheibler:2019vx} & \cmark & \cmark & \eref{cost_overiva} & {\hrulefill~IP~\hrulefill} & {\hrulefill~Parametric~\hrulefill} & \\
    OverIVA-IP2  & \xmark & \cmark &  \eref{cost_overiva} & {\hrulefill~IP2~\hrulefill}  & {\hrulefill~Parametric~\hrulefill} & $\star$New$\star$ \\
    OverIVA-IP-NP & \cmark & \cmark & \eref{cost_overiva} & {\hrulefill~IP~\hrulefill} & {\hrulefill~Non-parametric~\hrulefill} & $\star$New$\star$ \\
    OverIVA-IP2-NP & \xmark & \cmark &  \eref{cost_overiva} & {\hrulefill~IP2~\hrulefill}  & {\hrulefill~Non-parametric~\hrulefill} & $\star$New$\star$ \\
    OverIVA-Demix/BG & \cmark & \cmark & \eref{cost_overiva} & \multicolumn{2}{c}{\hrulefill~Joint~\hrulefill} & $\star$New$\star$ \\
    FIVE~\cite{Scheibler:2019tt} & \cmark & \xmark & \eref{cost_overiva} & \multicolumn{2}{c}{\hrulefill~Joint~\hrulefill} & \\
    OGIVEs~\cite{Koldovsky:fn} & \cmark & \xmark & \eref{cost_overiva} & {\hrulefill~Gradient~\hrulefill} & {\hrulefill~Parametric~\hrulefill} & \\
    AuxIVA-IP~\cite{Ono:2011tn}  & \cmark & \cmark & \eref{cost_iva} & {\hrulefill~IP~\hrulefill} & $\varnothing$ & Outputs $K$ largest sources \\
    AuxIVA-IP2~\cite{Ono:2020iva} & \cmark & \cmark &  \eref{cost_iva} & {\hrulefill~IP2~\hrulefill}  & $\varnothing$ & Outputs $K$ largest sources \\
    \bottomrule
  \end{tabular}
  \caption{Summary of the algorithms compared in the experiments}
  \tlabel{algo_ref}
\end{table*}

\section{Performance Evaluation}
\seclabel{perfeval}

The performance evaluation aims at assessing the following properties of the algorithms
\begin{itemize}
  \item Separation performance for speech signals
  \item Runtime characteristics
  \item Success rate of the separation under different conditions
\end{itemize}
The evaluation is done through numerical experiments.
We compare the following algorithms.
\begin{enumerate}
  \item \textbf{OverIVA-IP}~\cite{Scheibler:2019vx}:
    Sources are updated one-by-one with the IP rules~\cite{Ono:2011tn}.  
    The parametric background update from \sref{param_bg_update} (equation~\eref{J_up}) is applied before each source update.
  \item \textbf{OverIVA-IP2 (new)}:
    Sources are updated two-by-two with the IP2 rules~\cite{Ono:2020iva}.  
    The parametric background update from \sref{param_bg_update} (equation~\eref{J_up}) is applied before each source update.
  \item \textbf{OverIVA-IP-NP (new)}:
    Sources are updated one-by-one with the IP rules~\cite{Ono:2011tn}.  
    The non-parametric background update from \sref{non_param_bg_update} is applied before each source update.
  \item \textbf{OverIVA-IP2-NP (new)}:
    Sources are updated two-by-two with the IP2 rules~\cite{Ono:2020iva}.  
    The non-parametric background update from \sref{non_param_bg_update} is applied before each source update.
    See \algref{overiva-ip2-np}.
  \item \textbf{OverIVA-DX/BG (new)}:
    The sources are updated one-by-one, but always jointly with the background by applying Theorem~\ref{thm:jisamm_three_plus}.
    See \algref{overiva-dx/bg}.
  \item \textbf{FIVE}~\cite{Scheibler:2019tt}: The algorithm is the same as OverIVA-DX/BG, but specialized for a single source. In that case, Theorem~\ref{thm:jisa_L2} applies.
  \item \textbf{OGIVEs}~\cite{Koldovsky:fn}: The gradient-based OGIVE algorithm with the switching criterion.
  \item \textbf{AuxIVA-IP}~\cite{Ono:2011tn}: After full IVA with IP, the $K$ strongest sources are selected at the output.
  \item \textbf{AuxIVA-IP2}~\cite{Ono:2020iva}: After full IVA with IP2, the $K$ strongest sources are selected at the output.
\end{enumerate}
Note that not all algorithms apply to all cases.
FIVE and OGIVEs only apply in the single source case.
OverIVA-IP2 and OverIVA-IP2-NP always extract at least two sources.
The algorithms and their properties are summarized in \tref{algo_ref}.
All use the non-linearity corresponding to the Laplace distribution, namely $\varphi(r) = \frac{1}{2r}$.

The algorithms operate in the time-frequency domain and are preceded by a 4096 points STFT with $3/4$-overlap.
We use a Hamming window and the matching synthesis window for optimal reconstruction.
Speech samples of approximately \SI{20}{\second} are created by concatenating utterances from the CMU Sphinx database \cite{Kominek:2004vf}.
The simulation is conducted at a sampling frequency of \SI{16}{\kilo\hertz}.
The reverberation is simulated using the \texttt{pyroomacoustics} Python package~\cite{Scheibler:2018di}.

We found that pre-whitening the signals by applying a principal component analysis (PCA) prior to separation is always beneficial.
The input to all algorithms is thus the output of the PCA.
The scale of the separated signals is restored by minimizing the distortion with respect to the first microphone~\cite{Matsuoka:2002da}.

For the evaluation, we use the SI-SDR and SI-SIR~\cite{LeRoux:2018tq}.
The SI-SDR measures how different the extracted signal is from the groundtruth.
The SI-SIR measures how much of the other sources and noise is left in the extracted signal.
We also measure the success rate of the different algorithms that we define as the probability that the output SI-SIR is larger than~\SI{0}{\decibel}.
That is, the power of the target signal in the output is larger than the combined power of the rest.

\begin{figure*}
  \centering
  \includegraphics{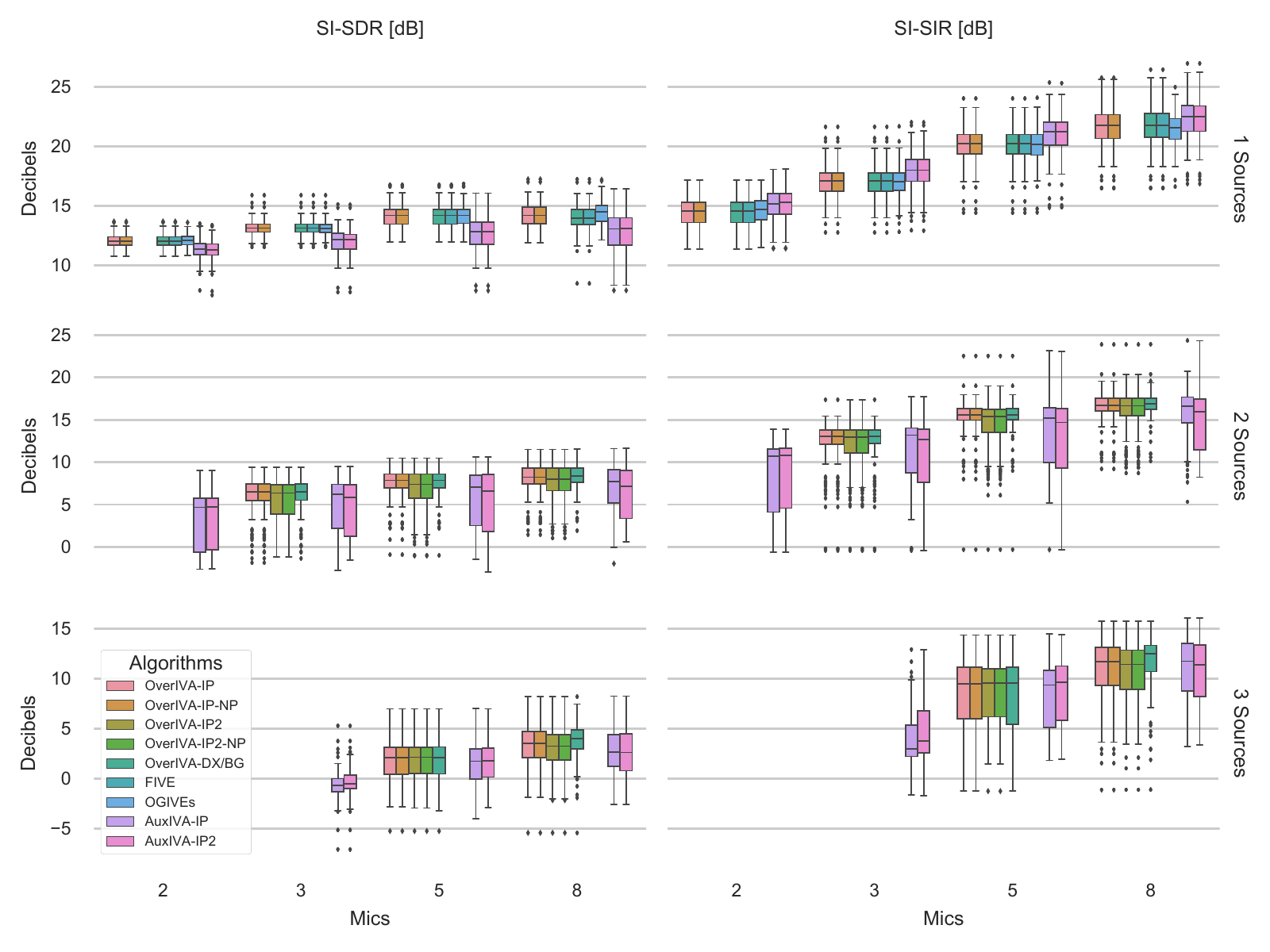}
  \caption{
    Box-plots of SI-SDR (left) and SI-SIR (right) after processing for the different algorithms.
    The number of sources increases from 1 to 3 from top to bottom.
    The number of microphones increases with the horizontal axis.
  }
  \flabel{sep_perf}
\end{figure*}

\subsection{Separation and Runtime Performance}

\subsubsection{Setup}

We simulate 100 random rectangular rooms with walls between \SI{6}{\meter} and \SI{10}{\meter} and ceiling from \SI{2.8}{\meter} to \SI{4.5}{\meter} high.
Simulated reverberation times ($T_{60}$) range from \SI{60}{\milli\second} to \SI{450}{\milli\second}.
See \ffref{experiment_rt60} for a histogram of the $T_{60}$.
Sources and microphone array are placed at random at least \SI{50}{\centi\meter} away from the walls and between \SI{1}{\meter} and \SI{2}{\meter} high.
The array is circular and regular with 2, 3, 5, or 8 microphones, and radius such that neighboring elements are \SI{2}{\centi\meter} apart.
All sources are placed further from the array than the critical distance of the room --- the distance where direct sound and reverberation have equal energy.
It is computed as $d_{\text{crit}} = 0.057\sqrt{V/T_{60}}\,\si{\meter}$, with $V$ the volume of the room~\cite{Kuttruff:2009uq}.
We let the target source be the one closest to the array, in the interval $[d_{\text{crit}}, d_{\text{crit}}+1]$.
The $Q=10$ interferers are at least $(d_{\text{crit}} + 1)\,\si{\meter}$ from the array.
We define 
\begin{equation}
  \SINR = \frac{K \sigma_T^2}{Q \sigma_I^2 + \sigma_w^2},
  \elabel{sinr}
\end{equation}
where $\sigma_T^2$ and $\sigma_I^2$ are the variance of target and interferers, respectively, at the first microphone.
The uncorrelated noise variance $\sigma_w^2$ is set to be \SI{1}{\percent} of the total noise-and-interference.
In this experiment, we fix $\SINR=$~\SI{10}{\decibel}.
MM algorithms are run for 100 iterations.
OGIVEs is run for 2000 iterations with step size of 0.1.

\begin{figure*}
  \centering
  \begin{subfigure}{\textwidth}
    \centering
    \includegraphics{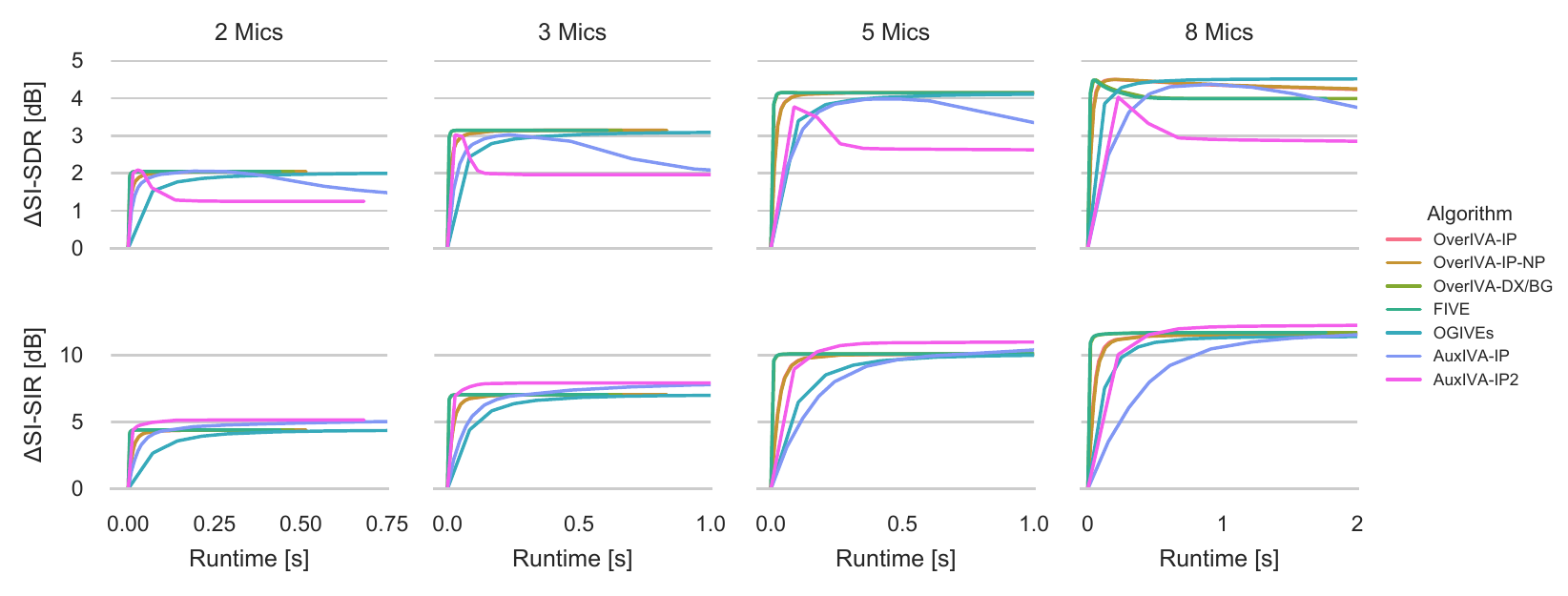}
    \caption{One target source}
  \end{subfigure}
  \begin{subfigure}{\textwidth}
    \centering
    \includegraphics{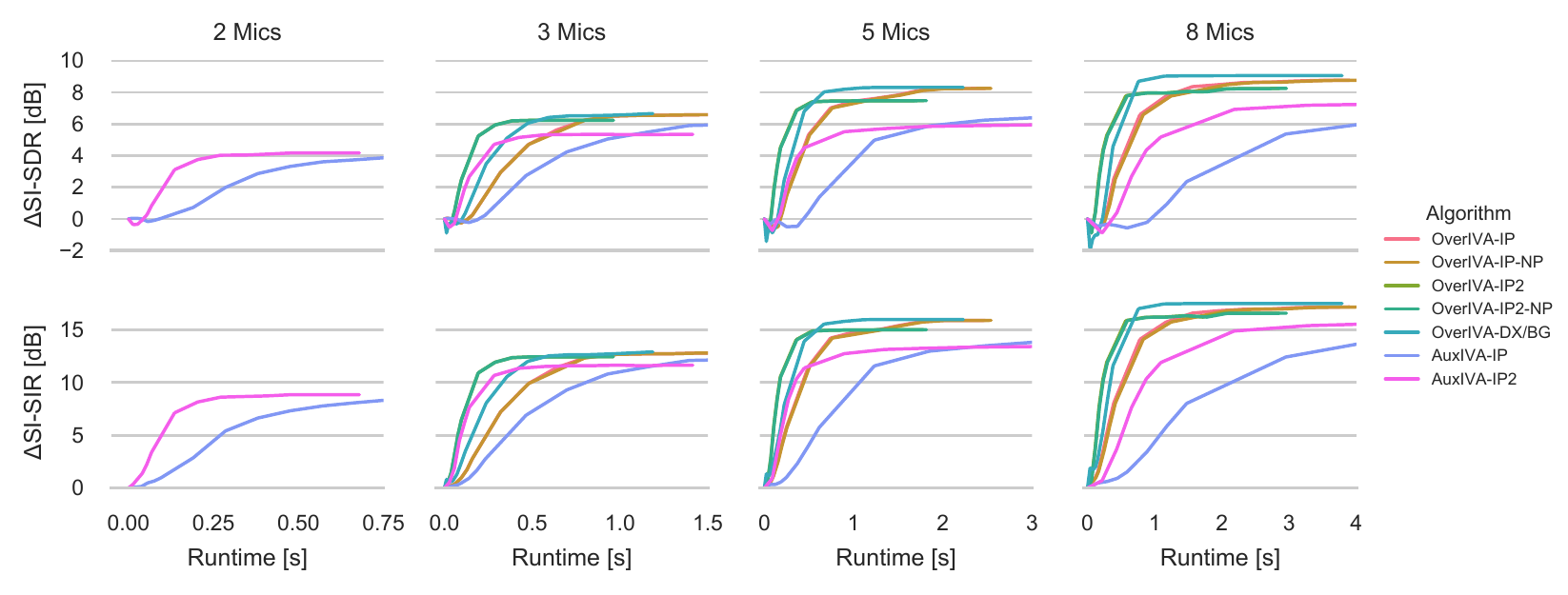}
    \caption{Two target sources}
  \end{subfigure}
  \begin{subfigure}{\textwidth}
    \centering
    \includegraphics{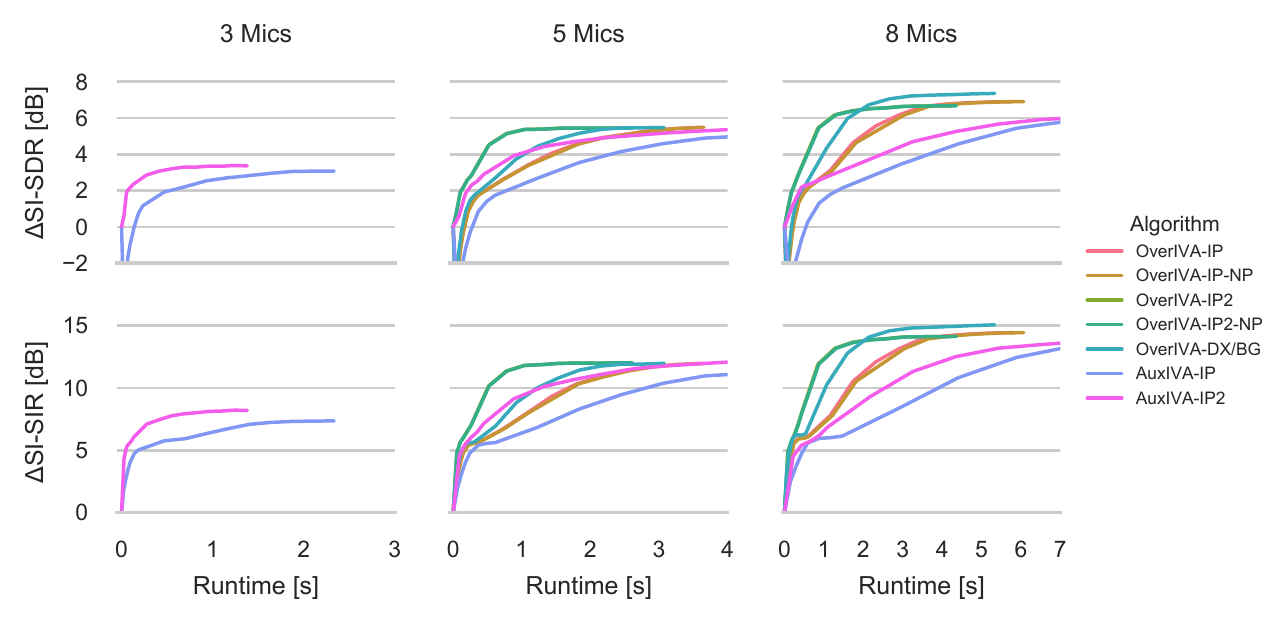}
    \caption{Three target sources}
  \end{subfigure}
  \caption{The SI-SDR and SI-SIR improvement as a function of wall-clock time for different algorithms and for (a) single, (b) two, and (c) three target sources. From left to right we use 2, 3, 5, and 8 microphones. In each sub-graph, the top and bottom row show SDR and SIR, respectively.}
  \flabel{runtime_perf}
\end{figure*}

\subsubsection{Separation Performance}

\ffref{sep_perf} shows box-plots of the SI-SDR and SI-SIR after the processing.
The end of the whiskers are placed at the largest sample smaller than $1.5\times$ the inter-quantile range.
Dots are samples outside this range.

First, we can evaluate the performance of the determined case in this experiment.
For two sources/two microphones and three sources/three microphones, we only evaluated AuxIVA-IP and AuxIVA-IP2 as the number of dimensions is not enough to include the background model.
We observe that there is little difference in the final values of both methods.
In contrast, we observe that, overall, adding more microphones steadily improves the separation performance, regardless of the algorithm used.
This makes a strong case for working in the overdetermined regime.
The improvement is especially visible in the SI-SIR, meaning the sources are better separated with more microphones.

For a single source, AuxIVA-IP/IP2 has lower SI-SDR, but higher SI-SIR than overdetermined methods.
We conjecture this might be due to the better match of the Gaussian background in this experiment.
For two and three sources, overdetermined methods perform better overall and in all cases they have smaller variance.
While all overdetermined methods perform fairly close to each others, we will note two things.
First, for two sources, the OverIVA-IP2 and OverIVA-IP2-NP have a slightly larger variance than the others.
As shown, in convergence results we show next, these methods are very fast.
However, it seems they might get stuck in local minimas in some cases.
Second, we notice that OverIVA-DX/BG performs better than the other methods by a small but noticeable margin for three sources and 8 microphones.
This might make it the best method when the number of microphones is large.

Finally, there is no difference of performance between the parametric and non-parametric background updates.
Thus, the parametric updates should be preferred as they are more economical in computation and memory.

\subsubsection{Convergence and Runtime Performance}

The average runtime of each algorithm normalized for \SI{1}{\second} of audio input was measured and divided by the number of iterations.
In \ffref{runtime_perf}, we plot the improvement over time of the SI-SDR and SI-SIR from their initial value.

Starting by the determined case, we observe that owing to the pairwise updates, AuxIVA-IP2 is much faster than AuxIVA-IP.
In the overdetermined case, for a single source, FIVE is the unquestionable winner, with a striking convergence speed that was noted in the original paper~\cite{Scheibler:2019tt}.
OverIVA-DX/BG is essentially the same algorithm in this case and thus performs alike.
Next, and in order, are OverIVA-IP/OverIVA-IP-NP, AuxIVA-IP2, OGIVEs, AuxIVA-IP.
For two and three sources, OverIVA-IP2/OverIVA-IP2-NP are the fastest converging methods.
AuxIVA-IP2 is very competitive for two and three microphones, but is left behind for 5 and 8.
OverIVA-DX/BG is slightly slower than OverIVA-IP2 methods, but with little difference for 5 and 8 microphones.
However, it achieves a slightly larger final value.

\subsection{Investigation of Separation Success Under Different Background Conditions}

As discussed in \sref{discussion}, the selection of the target source relies entirely on the cost function.
Here, we investigate under what conditions the expected source is correctly recovered.
We define the separation as successful when the output SI-SIR is larger than \SI{0}{\decibel}.
We focus on three different parameters: SINR, Gaussianity of the background, and reverberation time.
The SINR is set according to \eref{sinr} at \SI{-5}{\decibel}, \SI{0}{\decibel}, \SI{5}{\decibel}, and \SI{10}{\decibel}.
The number of interferers goes from one to ten.
This also roughly corresponds to the Gaussianity of the background.
For a single interferer, the background is very non-Gaussian, while for ten, it is close to Gaussian.
All interferers are placed beyond the critical distance. 
Sources closer than the critical distance experience a shorter reverberation time.
We compare placements of target sources from \SI{20}{\centi\meter} up to the critical distance of \SI{2}{\meter}.
We do not include OverIVA with non-parametric updates as it was shown to perform identically to parametric updates in the previous experiment.
Similarly, we do not include AuxIVA-IP, as AuxIVA-IP2 has been shown to be superior.

\begin{figure}
  \centering
  \begin{subfigure}{\linewidth}
    \centering
    \includegraphics{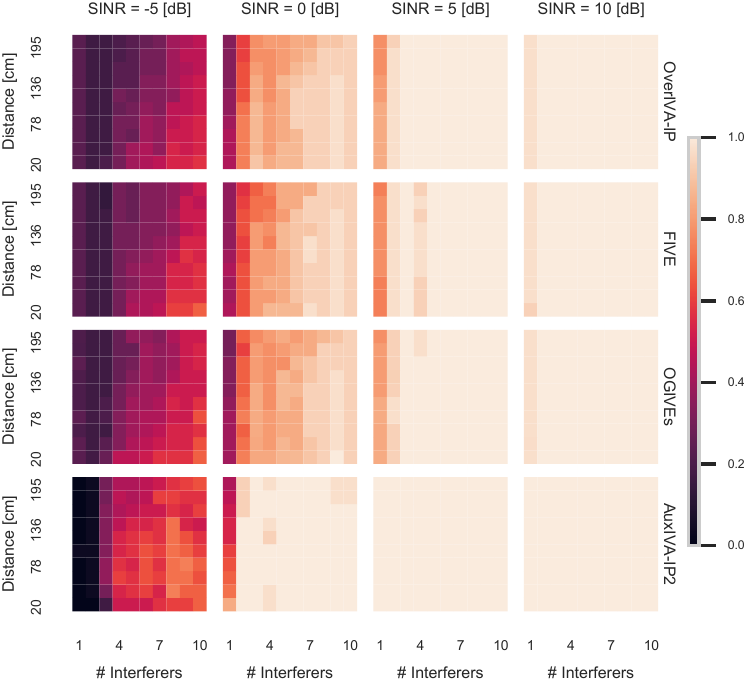}
    \caption{One target source}
    \flabel{success_1tgt}
  \end{subfigure}
  \begin{subfigure}{\linewidth}
    \centering
    \vspace{0.2cm}
    \includegraphics{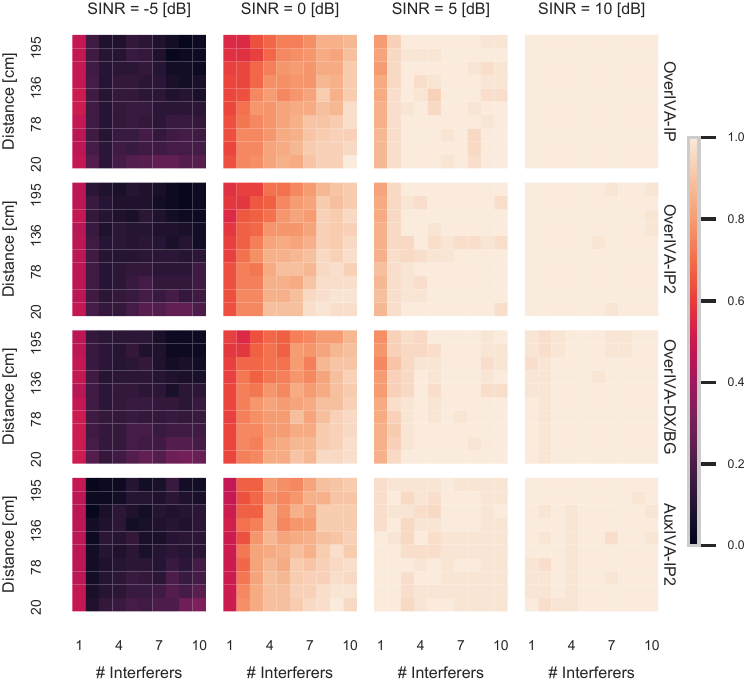}
    \caption{Two target sources}
  \end{subfigure}
  \begin{subfigure}{\linewidth}
    \vspace{0.2cm}
    \centering
    \includegraphics{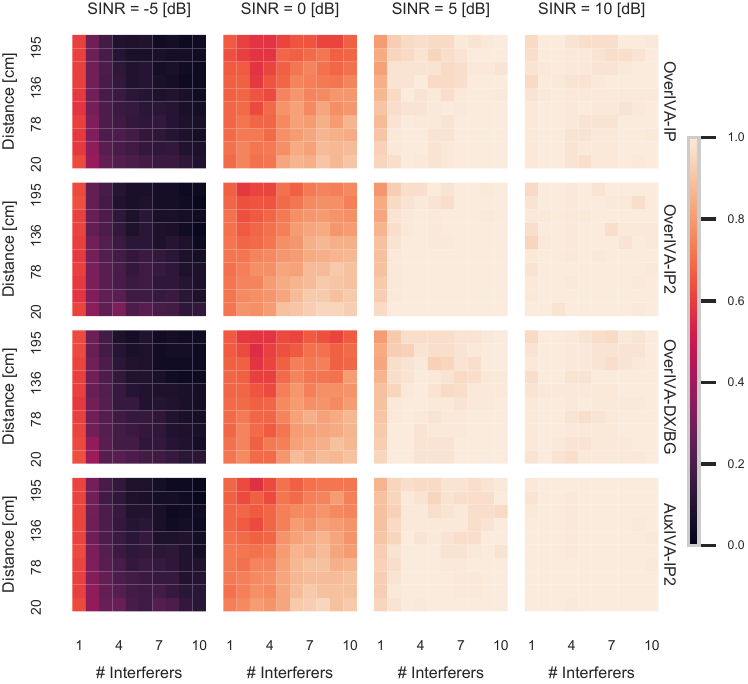}
    \caption{Three target sources}
  \end{subfigure}
  \caption{Probability of success of algorithms as a function of the number of interferers and distance to the microphones. Each row shows a single algorithm labeled on the right of the plot. Columns are for increasing SINR.}
  \flabel{success}
\end{figure}

\subsubsection{Setup}

We simulate a \SI{9}{\meter}$\times$\SI{12}{\meter}$\times$\SI{4.5}{\meter} room with reverberation time of \SI{415}{\milli\second}.
We use an array with seven microphones --- one at the center and six placed uniformly on a circle of radius \SI{2}{\centi\meter}.
The array is placed at (\SI{4.496}{\meter}, \SI{5.889}{\meter}, \SI{2.327}{\meter}), i.e., a little bit off the center of the room.
All target sources are placed equidistant from the array center and separated by equal angles.
Interferers are placed at random in the room, but beyond the critcal distance, i.e., $\sim$\SI{2}{\meter}.
We repeat the experiment for 30 different placements of interferers and noise patterns.
All MM algorithms are run for 50 iterations and OGIVEs is run for 1000.

\subsubsection{Results}

\ffref{success} shows the empirical success probability for the algorithms in different conditions.
The SINR increases with columns from left to right.
Each row is for one algorithm labeled on the right.
Each cell contains a $10\times 10$ color map with each pixel corresponding to a combination of distance and number of interferers.
Darker and lighter colors indicate low and high probabities of success, respectively.

We find the dominant factor for successful recovery to be the SINR.
It needs to be positive for reliable source recovery.
For a single target source AuxIVA-IP2 is most reliable with success probability close to one whenever SINR$\geq$\SI{0}{\decibel}, except for a single interferer at \SI{0}{\decibel}.
In the latter case, the algorithm is expected to be wrong half the time since there are only two sources of the same power.
For overdetermined algorithms for a single source at \SI{0}{\decibel}, the reliability increases with the number of interferers.
Separation is reliable at \SI{5}{\decibel} and over two interferers, or \SI{10}{\decibel} and two or more interferers.
For two and three sources, AuxIVA-IP2 loses its edge and behaves very similarly to overdetermined algorithms.
Overall the more sources to separate there is, the harder the problem gets.
This is reflected by a decreasing success rate at \SI{0}{\decibel}.

Overall, the effect of the source distance is small, but noticeable.
Sources closer to the microphones are more likely to be successfully separated.
One can also notice that for two and three sources, SINR \SI{-5}{\decibel}, and a single interferer, the success probability is larger than for more interferers.
This happens because, in this case, by a simple combinatorial argument, we are guaranteed to include some of the correct sources in the output.

\section{Conclusion}
\seclabel{conclusion}

We introduced JISA-MM, a framework for joint independent subspace analysis based on majorization-minimization optimization.
This framework applies to super-Gaussian contrast functions, and we derive several update rules leading to efficient iterative algorithms to minimize the negative log-likelihood of the observed signals.
The resulting algorithms are hyperparameter free and easy to implement in practice.
We further apply JISA-MM to the BSE and BMSE, whereas a few target sources are to be separated from a Gaussian background.
We show how some existing algorithms, as well as new ones, can be derived from our general framework.
In numerical experiments, we compare all these algorithms in terms of separation, speed, and robustness to model mismatch.
In the single source case, FIVE~\cite{Scheibler:2019tt} flatly beats all other methods in terms of convergence speed.
For two and three sources, OverIVA-IP2 and OverIVA-DX/BG are the strongest contender with the later being somewhat more reliable for larger numbers of microphones.

In terms of robustness, our experiments show that all the overdetermined methods require at least positive SINR to succeed.
In the single source case, a very mismatched background model also seems to negatively affect performance.
There, AuxIVA appears more robust at low SINR, but loses its advantage when going to two and three target sources.
The need for positive SINR seems to be a consequence of fully relying on the cost function for target separation.
We conclude that, without further guidance, the dominant sources are extracted.
Alternatively, several recent works propose to guide source extraction with spatial constraints~\cite{Brendel:2020wc}, speaker identification via x-vector~\cite{Jansky:2019wa}, or a pilot signal correlated to the source~\cite{Jansky:2020wn}.
All these methods may benefit from the efficient algorithms proposed in this paper.

\bibliographystyle{IEEEtran}
\bibliography{refs}

\ifCLASSOPTIONcaptionsoff
  \newpage
\fi


\begin{IEEEbiography}[{\includegraphics[width=1in,height=1.25in,clip,keepaspectratio]{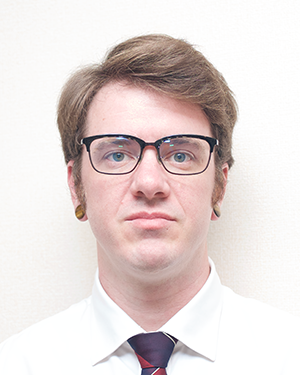}}]{Robin Scheibler} (M'07) is a specially appointed associate professor at the Tokyo Metropolitan University (Tokyo, Japan).
  Robin received his B.Sc, M.Sc, and Ph.D. from Ecole Polytechnique Fédérale de Lausanne (EPFL, Switzerland).
  He also worked at the research labs of NEC Corporation (Kawasaki, Japan) and IBM Research (Z\"{u}rich, Switzerland).
  From March 2020, he will be a researcher at LINE Corporation (Tokyo, Japan).
  Robin's research interests are in efficient algorithms for signal processing, and array signal processing more particularly.
  He also likes to build large microphone arrays and is the lead developer of pyroomacoustics, an open source library for room acoustics simulation and array signal processing.
\end{IEEEbiography}
\begin{IEEEbiography}[{\includegraphics[width=1in,height=1.25in,clip,keepaspectratio]{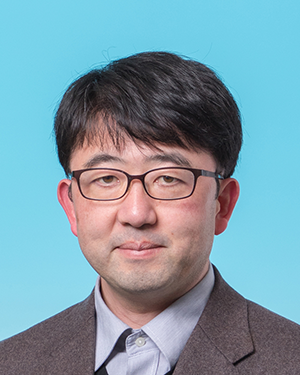}}]{Nobutaka Ono}
  (M'02--SM'13) received the B.E., M.S., and Ph.D degrees in Mathematical
  Engineering and Information Physics from the University of Tokyo, Japan, in
  1996, 1998, 2001, respectively. He joined the Graduate School of Information
  Science and Technology, the University of Tokyo, Japan, in Apr. 2001 as a
  Research Associate and became a Lecturer in Apr. 2005. He moved to the
  National Institute of Informatics, Japan, as an Associate Professor in Apr.
  2011 and became a Professor in Sep. 2017. He moved to Tokyo Metropolitan
  University in Oct. 2017. His research interests include acoustic signal
  processing, specifically, microphone array processing, source localization
  and separation, machine learning and optimization algorithms for them.
  He is
  the author or co-author of more than 240 articles in international journal
  papers and peer-reviewed conference proceedings. He was a Tutorial speaker at
  ISMIR 2010 and ICASSP 2018, a special session chair in EUSIPCO 2013, 2015,
  2017, 2018, and 2019, a chair of SiSEC (Signal Separation Evaluation
  Campaign) evaluation committee in 2013 and 2015. 
  He was an Associate Editor
  of the IEEE Transactions on Audio, Speech and Language Processing during 2012
  to 2015.He has been a member of IEEE Audio and Acoustic Signal Processing
  (AASP) Technical Committee since 2014. He is a senior member of the IEEE
  Signal Processing Society, and a member of the Acoustical Society of Japan
  (ASJ), the Institute of Electronics, Information and Communications Engineers
  (IEICE), the Information Processing Society of Japan (IPSJ), and the Society
  of Instrument and Control Engineers (SICE) in Japan. He received the Sato
  Paper Award and the Awaya Award from ASJ in 2000 and 2007, respectively, the
  Igarashi Award at the Sensor Symposium on Sensors, Micromachines, and Applied
  Systems from IEEJ in 2004, the best paper award from IEEE ISIE in 2008,
  Measurement Division Best Paper Award from SICE in 2013, the best paper award
  from IEEE IS3C in 2014, the excellent paper award from IIHMSP in 2014, the
  unsupervised learning ICA pioneer award from SPIE.DSS in 2015, the Sato Paper
  Award from ASJ and two TAF Telecom System Technology Awards in 2018, and Best
  Paper Award from APSIPA in 2018.
\end{IEEEbiography}





\end{document}